\documentclass[12pt]{amsart}

\usepackage[left=12mm, right=12mm, top=10mm, bottom=15mm]{geometry}

\usepackage[english]{babel}
\usepackage[T1]{fontenc}
\usepackage[utf8]{inputenc}

\theoremstyle{plain}
    \newtheorem{theorem}{Theorem}[section]
    \newtheorem{lemma}[theorem]{Lemma}
    \newtheorem{proposition}[theorem]{Proposition}

\theoremstyle{definition}

    \newtheorem{example}[theorem]{Example}

\usepackage[
    draft = false,
    unicode = true,
    colorlinks = true,
    allcolors = blue,
    hyperfootnotes = true,
    citecolor = red
]{hyperref}
\usepackage{amsfonts,amsmath,amssymb,amscd,amsthm,url}
\usepackage[inline]{enumitem}
\usepackage{graphicx,epsf,afterpage, wrapfig}
\usepackage{soul}
\usepackage{multicol}
\usepackage{array}
\usepackage{braket}
\usepackage{epigraph}
\usepackage{rotating, floatflt}
\usepackage{xstring} 

\usepackage[
backend=biber,
style=numeric-comp,
sorting=none,
giveninits=true
]{biblatex}
\usepackage{csquotes}

\usepackage{xcolor}

\usepackage{ulem}

\usepackage{tikz}
\usetikzlibrary{positioning, calc, intersections, through, arrows, matrix, chains, math}
\usetikzlibrary{arrows.meta}
\usetikzlibrary{shadows}
\usetikzlibrary{graphs}
\usetikzlibrary{graphs.standard}
\usetikzlibrary{patterns}
\usetikzlibrary{decorations.pathreplacing,calligraphy,backgrounds}

\DeclareFieldFormat{usera}{\href{https://arxiv.org/abs/#1}{\texttt{arXiv:#1}}}

\renewbibmacro*{finentry}{\printfield{usera}\newunit\finentry}

\renewbibmacro{in:}{%
  \iffieldundef{usera}
    {\printtext{\bibstring{in}\intitlepunct}}
    {}%
}

\newcommand\norm[1]{\ensuremath{\left\lVert#1\right\rVert}}
\newcommand\abs[1]{\ensuremath{\left\lvert#1\right\rvert}}

\newcommand{\Pcal}{\ensuremath{\mathcal{P}}}

\DeclareMathOperator{\supp}{supp}

\newcommand{\Bcal}{\mathcal{B}}

\newcommand{\Ical}{\mathcal{I}}

\newcommand{\Mcal}{\mathcal{M}}
\newcommand{\Hcal}{\mathcal{H}}

\newcommand{\Fcal}{\mathcal{F}}

\newcommand{\Ucal}{\mathcal{U}}

\newcommand{\Abf}{{\ensuremath{\mathbf{A}}}}

\newcommand{\Dbf}{{\ensuremath{\mathbf{D}}}}

\newcommand{\Ibf}{{\ensuremath{\mathbf{I}}}}

\newcommand{\Pbf}{{\ensuremath{\mathbf{P}}}}
\newcommand{\Qbf}{{\ensuremath{\mathbf{Q}}}}

\newcommand{\Sbf}{{\ensuremath{\mathbf{S}}}}

\newcommand{\Ubf}{{\ensuremath{\mathbf{U}}}}

\newcommand{\Phbf}{{\ensuremath{\mathbf{\Phi}}}}

\newcommand{\defing}[1]{\textbf{\emph{\mathversion{bold}#1}}}

\newcommand{\continuum}{\ensuremath{\mathfrak{c}}}

\newcommand{\R}{\ensuremath{\mathbb{R}}}
\newcommand{\Z}{\ensuremath{\mathbb{Z}}}
\newcommand{\Cx}{\ensuremath{\mathbb{C}}}

\newcommand{\N}{\ensuremath{\mathbb{N}}}

\renewcommand{\geq}{\geqslant}

\renewcommand{\leq}{\leqslant}

\newcounter{mcnt}

\newcounter{wordcnt}

\begin{document}

\title{Singularising preserving countable additivity \\ quantum channels on quantum measurable cardinals}

\author{Sviatoslav V. Dzhenzher}

\begin{abstract}
    We investigate the structural and dynamical properties of a class of quantum channels on von Neumann algebras induced by averaging over operator groups via the Pettis integral.
     Utilising the classical Yosida--Hewitt decomposition, we focus on the interplay between the set-theoretic properties of the underlying space and the topological nature of the resulting quantum states.
    We establish a sufficient condition under which a channel preserves $\sigma$-additivity and exhibits a singularising property, completely suppressing the normal component of any incoming state.
    In conjunction with the theory of Ulam real-valued measurable cardinals, this framework reveals a novel phenomenon: the existence of quantum channels that transform normal states into singular yet strictly $\sigma$-additive states.
    Furthermore, we analyse the structural constraints on the preservation of state purity imposed by the cardinality of the continuum, and extend our constructions to invariant measures on groups and their unitary representations, establishing the convergence of their Ces\`aro averages in the strong operator topology.
\end{abstract}

\thanks{\hspace{-4mm}
S.\,V. Dzhenzher: sdjenjer@yandex.ru. orcid: 0009-0008-3513-4312
\\
Moscow Institute of Physics and Technology 141701, Institutskii per. 9, Dolgoprudny, Russia}

\maketitle
\thispagestyle{empty}


\vspace{5mm}

\emph{MSC2020 Subject Classification: Primary 03E55, 46L30; Secondary 81P47, 46G10.}

\section{Introduction}

The theory of quantum channels and the evolution of quantum states on von Neumann algebras constitute the mathematical foundation of modern quantum mechanics and quantum information theory \cite{Takesaki1979, Kholevo2011-ch2}. In physical applications, primary attention is traditionally devoted to normal states, which possess the property of complete additivity and are associated with density matrices in separable Hilbert spaces. However, in infinite-dimensional systems, especially when transitioning to spaces of uncountable dimension, singular states begin to play a significant role \cite{Akemann-1967-sequential, Amosov-Sakbaev-2013}. These linear functionals vanish on the ideal of compact operators \cite{Farah-Wofsey} and lack a direct physical analogue in the form of density matrices \cite{Amosov-Bikchentaev-Sakbaev-24, Amosov-Sakbaev-2025}.
See also \cite{Busovikov-Sakbaev-26} for the dynamics of singular quantum states generated by the averaging of random shifts over some particular measures.

According to the classical Yosida--Hewitt decomposition \cite{Yosida-Hewitt}, any quantum state can be uniquely represented as a convex combination of its normal and singular components. While the dynamics of normal states have been studied in detail \cite{Amosov-Sakbaev-2025}, the mechanisms of generation and transformation of singular states under the action of quantum channels remain less explored. Of particular theoretical interest are channels possessing the \emph{singularising} property, i.e., those that map any initial state to a singular one \cite{DzhenzherDzhenzherSakbaev26, Dzhenzher26-mc}.

In the present paper, we investigate a class of quantum channels induced by averaging over groups of operators via the Pettis integral with respect to a given measure. The main focus is on the interplay between the set-theoretic properties of the underlying space (in particular, the properties of Ulam real-valued measurable cardinals \cite{Jech-p1-ch10}) and the topological properties of the resulting quantum states, building upon the framework of quantum measurable cardinals established in \cite{Blecher-Weaver}.

The main results of this paper are summarised as follows:
\begin{itemize}
    \item \textbf{\mathversion{bold}Preservation of $\sigma$-additivity:} we prove that if a quantum channel is constructed as the Pettis integral over a $\sigma$-additive measure defined on all subsets of some cardinal, then it strictly preserves the $\sigma$-additivity of quantum states (Theorem~\ref{t:pres-sig-add}).
    \item \textbf{Criterion for singularisation:} we establish a sufficient condition for a channel to be singularising. Specifically, if the $\sigma$-additive measure vanishes on singletons, the corresponding channel completely suppresses the normal component of any state (Theorem~\ref{t:sing-out}). Combined with the theory of Ulam measurable cardinals, this leads to a fascinating phenomenon: the existence of channels that convert normal states into singular yet strictly $\sigma$-additive states.
    \item \textbf{Preservation of purity:} we examine the conditions under which the considered channels can preserve the purity of vector states and establish structural limitations imposed by the cardinality of the continuum $\mathfrak{c}$ (Section~\ref{s:purity}).
    \item \textbf{Dynamics on groups:} we generalise the proposed construction to the case of left-invariant measures on groups and their unitary representations. We study the ergodic properties of these channels, characterise their fixed points, and prove the convergence of their Ces\`aro averages in the strong operator topology (Section~\ref{s:unitary}).
\end{itemize}

The structure of the paper is organised as follows.
In Section~\ref{s:von-Neumann}, we recall the necessary preliminaries regarding quantum states on von Neumann algebras, normality, singularity, and the Yosida--Hewitt decomposition.
Section~\ref{s:singularisation} is devoted to the construction of quantum channels on $\ell^2(\kappa)$ spaces and the proofs of their key properties regarding $\sigma$-additivity and singularisation.
In Section~\ref{s:purity}, we analyse the geometry of the state space and the behaviour of pure states under these channels using excision criteria.
In Section~\ref{s:unitary}, we consider analogous constructions for measures on general groups and their unitary representations, and investigate the asymptotic dynamics of quantum channel iterations.

\section{Quantum states on von Neumann algebras}\label{s:von-Neumann}

Let $\Hcal$ be a Hilbert space.
We presume that the inner product \((\cdot,\cdot)\) is linear by the second argument, as in quantum mechanics.
Since we usually work with normed states, it is convenient to use the notation \(S_1(\Hcal)\subset\Hcal\) for the unit sphere in $\Hcal$; that is, \(u \in S_1(\Hcal)\) if \(\norm{u}=1\).

Denote by \(\Bcal(\Hcal)\) the Banach algebra of linear bounded operators \(\Hcal\to\Hcal\).
We denote by \(\Pcal(\Hcal)\subset\Bcal(\Hcal)\) the family of \defing{projectors}.
For a normed \(u\in S_1(\Hcal)\), denote by \(\Pbf_u \in \Pcal(\Hcal)\) the \defing{projector} on the subspace generated by $u$; that is, for any \(v\in\Hcal\) given by
\[
    \Pbf_uv = (u,v)u,\qquad
    (v, \Pbf_u v) = \abs{(u,v)}^2.
\]
The \emph{identity operator} \(\Ibf\in\Bcal(\Hcal)\) is the projector onto the whole $\Hcal$.

Let \(\Mcal\subset\Bcal(\Hcal)\) be a von Neumann algebra of \defing{observables}; that is, a weakly closed $*$-algebra of $\Bcal(\Hcal)$ containing $\Ibf$.
For example, such algebras could be
\begin{itemize}
    \item the whole algebra \(\Bcal(\Hcal)\),

    \item the algebra of diagonal operators on \(\ell_2\),

    \item the algebra \(\{\lambda\Ibf : \lambda\in\Cx\}\).
\end{itemize}

Denote by \(\Sigma(\Mcal) := S^+_1(\Mcal^*)\) the space of \defing{quantum states}, that is, the linear functionals \(\Mcal \to \Cx\), lying on the intersection of the positive cone and the unit sphere.
We denote the \defing{action} of a quantum state \(\rho \in \Sigma(\Mcal)\) on an observable \(\Abf\in \Mcal\) by
\[
    \braket{\rho, \Abf}.
\]

Recall that for a collection \(\{r_i:i\in I\}\) of non-negative real numbers, their sum is defined as
\[
    \sum_{i \in I}r_i = \sup \left\{\sum_{i \in E} r_i : \text{$E\subset I$ is finite}\right\}.
\]
If the sum is finite, then there is at most a countable collection of non-zero $r_i$.

A state \(\rho\in\Sigma(\Mcal)\) (or just a functional $\rho\in \Mcal^*$) is \defing{normal} if it is weakly-$*$ continuous \cite{Blecher-Weaver}. We denote the space of normal states by \(\Sigma_n(\Mcal)\).
It is known \cite{Blecher-Weaver, Takesaki1979} that a state \(\rho\in\Sigma(\Mcal)\) is normal if and only if it is \defing{completely additive}, which means that the equality
\[
    \Braket{\rho, \sum_\alpha \Pbf_\alpha} = \sum_\alpha \Braket{\rho, \Pbf_\alpha}
\]
holds for any family of pairwise orthogonal projectors \(\Pbf_\alpha\in\Pcal(\Hcal)\cap\Mcal\).
It is also known that a state \(\rho\in\Sigma(\Bcal(\Hcal))\) is normal if and only if
\[
    \sup \Braket{\rho, \Pbf} = 1,
\]
where the supremum is taken over projectors onto finite-dimensional subspaces. This fact is well known in the case when $\Hcal$ is separable. We present Proposition~\ref{p:normal-sup} below just to rid ourselves of any doubts, since we do not know the reference for it.

A state \(\rho\in\Sigma(\Mcal)\) is \defing{singular} \cite{Akemann-1967-sequential}, if there does not exist a positive normal non-zero functional \(\rho_n\in\Mcal^*\) such that \(\rho\geq\rho_n\).
It is well known, see for example \cite[Proposition~II.1]{Akemann-1967-sequential},
\cite[Theorem~3.8]{Takesaki1979}, that \(\rho\) is singular, if and only if for any non-zero projector \(\Pbf\in\Mcal\), there exists a non-zero projector \(\Pbf_0 \leq \Pbf\) (from $\Mcal$) such that
\[
    \Braket{\rho, \Pbf_0} = 0.
\]
We denote by \(\Sigma_s(\Mcal)\) the space of singular states.
Since we will usually work with \(\Mcal=\Bcal(\Hcal)\), it is convenient to keep in mind the following equivalent definition \cite[Definition~4.19]{Farah-Wofsey}: a state \(\rho\in\Sigma(\Bcal(\Hcal))\) is singular if and only if it vanishes on compact operators.

The well known result \cite[Theorem~2.14]{Takesaki1979}, which is a generalisation of Yosida--Hewitt decomposition \cite{Yosida-Hewitt}, states that any \(\rho\in\Sigma(\Mcal)\) can be represented as
\[
    \rho = \lambda\rho_n + (1-\lambda)\rho_s
\]
in the unique way for some \(\lambda\in[\,0,1\,]\), and states \(\rho_n\in\Sigma_n(\Mcal)\) and \(\rho_s\in\Sigma_s(\Mcal)\).

Now, we are ready to prove the promised criterion of normality.

\begin{proposition}\label{p:normal-sup}
    Let $\Hcal$ be a Hilbert space.
    A state \(\rho\in\Sigma(\Bcal(\Hcal))\) is normal if and only if
    \[
        \sup \Braket{\rho, \Pbf_{\text{finite}}} = 1,
    \]
    where the supremum is taken over projectors $\Pbf_{\text{finite}}$ onto finite-dimensional subspaces.
\end{proposition}

\begin{proof}
    Suppose \(\rho\) is normal. Then it is completely additive. Hence, taking the orthonormal basis \(\{e_\alpha\}\) in \(\Hcal\), we have
    \[
        1 = \Braket{\rho, \Ibf} = \sum_\alpha \Braket{\rho, \Pbf_\alpha} = \sup \Braket{\rho, \Pbf_{\text{finite}}},
    \]
    where the supremum is taken over projectors $\Pbf_{\text{finite}}$ onto finite-dimensional subspaces, the first equality follows from the normalisation of states, the second equality follows from complete additivity, and the last equality follows from the definition of the sum.

    Conversely, suppose that
    \[
        \sup\Braket{\rho, \Pbf_{\text{finite}}} = 1,
    \]
    where the supremum is taken over projectors $\Pbf_{\text{finite}}$ onto finite-dimensional subspaces.
    Take the decomposition
    \[
        \rho = \lambda\rho_n + (1-\lambda)\rho_s,
    \]
    where \(\rho_n\) is normal and $\rho_s$ is singular.
    Hence
    \[
        \Braket{\rho, \Pbf_{\text{finite}}} = \lambda\Braket{\rho_n, \Pbf_{\text{finite}}},
    \]
    and so
    \[
        1 = \sup\Braket{\rho, \Pbf_{\text{finite}}} = \lambda\sup\Braket{\rho_n, \Pbf_{\text{finite}}} = \lambda.
    \]
    Thus, $\rho=\rho_n$ is normal.
\end{proof}

A state \(\rho\in\Sigma(\Mcal)\) is said to be \defing{pure} if it is an extreme point of \(\Sigma(\Mcal)\).

For \(u\in S_1(\Hcal)\), a normal state \(\rho=\rho_u\in\Sigma_n(\Mcal)\) is called a \defing{vector state}, if for all \(\Abf\in\Mcal\) we have
\[
    \Braket{\rho, \Abf} = (u, \Abf u).
\]
Vector states are pure if \(\Mcal=\Bcal(\Hcal)\), and in some other cases.
See Example~\ref{ex:vec-non-pure} for the case when a vector state is not pure, and Proposition~\ref{p:excise-pv} for details.

All the notions above can be called classical.
The following notion, apparently, cannot be called classical, but it is of extreme interest to us.
A state \(\rho\in\Sigma(\Mcal)\) is \defing{$\sigma$-additive} (or \emph{countably additive}) if
\[
    \Braket{\rho, \sum_{n\in\N} \Pbf_n} = \sum_{n\in\N} \Braket{\rho, \Pbf_n}
\]
for any countable family of pairwise orthogonal projectors \(\Pbf_n\in\Mcal\).
From the above, it is clear that all normal states are $\sigma$-additive.

\section{Singularising quantum channels preserving countable additivity}\label{s:singularisation}

Fix any cardinal \(\kappa>\aleph_0\).
For a function \(f\colon\kappa\to\Cx\), we define its \defing{support} as
\[
    \supp f := \{i \in \kappa : f(i)\neq 0\}.
\]
As in \cite{Dzhenzher26-mc, Blecher-Weaver}, our primary focus will be on the Hilbert space \(\ell_2(\kappa)\) of functions \(f\colon \kappa\to\Cx\) with at most a countable support and the norm
\[
    \norm{f} := \sqrt{\sum_{i < \kappa}\abs{f(i)}^2}.
\]
We denote by \(\{e_i\}\) its orthonormal basis.

Let $\kappa$ additionally be equipped with an Abelian group structure.
For example, it can be done by taking the group \(\bigoplus_{i<\kappa} \Z_2\) of finite functions \(\kappa\to\Z_2\), which is equinumerous to \(\kappa\).
Or, for \(\kappa=\R\), it can be just the usual group \((\R,+)\).
Define the \defing{shift operator} \(\Sbf_j \in \Bcal(\Hcal)\) by the rule
\[
    \Sbf_j e_i := e_{i+j}.
\]

By $2^\kappa$ we denote the set of all subsets of~$\kappa$.

In this text, we consider finitely-additive measures \(\mu\colon2^\kappa\to[\,0,1\,]\) with \(\mu(\kappa)=1\).
Such measures always exist: for example, they can be delta-measures.
If $\kappa$ is a Ulam real-valued measurable cardinal, then \cite{Jech-p1-ch10} there exists a $\sigma$-additive measure that vanishes on singletons.

Recall that the integral over a finitely-additive measure \(\mu\colon2^\kappa\to[\,0,1\,]\) is defined as follows.
For a simple function \(f = \sum_{m=1}^M c_m\Ical\{f=c_m\}\colon\kappa\to\Cx\), the integral is defined as usual
\[
    \int f\,d\mu = \int f(j)\,d\mu(j) := \sum_{m=1}^M c_m \mu\{f=c_m\}.
\]
In general, a function $f\colon\kappa\to\Cx$ is integrable if there exists a sequence \(f_n\) of simple functions that is a Cauchy sequence in mean and converges to $f$ in measure.
In this case,
\[
    \int f\,d\mu = \int f(j)\,d\mu(j) := \lim_{n\to\infty} \int f_n\,d\mu.
\]
For details see, for example, \cite{Dunford-Schwartz-vol1}.

We say that a \defing{generalised quantum channel} is a linear positive map \(\Phbf\colon \mathcal{M}^* \to \mathcal{M}^*\) that preserves the state space, which means \(\Phbf(\Sigma(\mathcal{M})) \subset \Sigma(\mathcal{M})\).
Note that in standard quantum information theory, channels are usually defined as completely positive maps on the predual space \(\mathcal{M}_{*}\). However, since the main properties investigated in this paper~--- namely, \(\sigma\)-additivity preservation and singularisation~--- depend exclusively on the positivity and linearity of the mapping, we restrict our consideration to positive linear maps on the entire dual space \(\mathcal{M}^{*}\).
Since we consider only quantum states (but not functionals from $\Mcal^*$), we will denote generalised quantum channels as \(\Sigma(\Mcal)\to\Sigma(\Mcal)\).
Moreover, we will omit the word ``generalised'', and will restrict our attention on quantum channels \(\Sigma(\Bcal(\Hcal))\to \Sigma(\Bcal(\Hcal))\).
See classical exposition on quantum channels, for example, in \cite{Kholevo-qprobstat, Kholevo-qsys-chan-inf, Amosov-Sakbaev-2025}.

Note that the definition of quantum channels is given in the Schr\"{o}dinger representation.
There exists another picture, the Heisenberg representation, where the channels act on the operator algebra itself.
These two pictures are related via the equation
\[
    \Braket{\Phbf\rho, \Abf} = \Braket{\rho, \Phbf^*(\Abf)}.
\]

For a measure \(\mu\colon2^\kappa\to[\,0,1\,]\),
define the quantum channel \(\Phbf_\mu\colon\Sigma(\Bcal(\ell_2(\kappa)))\to\Sigma(\Bcal(\ell_2(\kappa)))\) given by the Pettis integral
\[
    \Phbf_\mu\rho := \int \Sbf_j\rho\Sbf_j^*\,d\mu(j).
\]
This means that for any \(\Abf\in\Bcal(\ell_2(\kappa))\),
\[
    \Braket{\Phbf_\mu\rho, \Abf} = \int \Braket{\Sbf_j\rho\Sbf_j^*,\Abf}\,d\mu(j) =
    \int \Braket{\rho,\Sbf_j^*\Abf\Sbf_j}\,d\mu(j).
\]
We will be mostly interested in quantum channels preserving $\sigma$-additivity: by this, we mean that these channels map initial $\sigma$-additive states to also $\sigma$-additive ones.

\begin{theorem}\label{t:pres-sig-add}
    Let $\mu\colon2^\kappa\to[\,0,1\,]$ be a $\sigma$-additive measure.
    Then the channel \(\Phbf_\mu\) preserves $\sigma$-additivity.
\end{theorem}

\begin{proof}
    Suppose that \(\rho\in\Sigma(\Bcal(\ell_2(\kappa)))\) is $\sigma$-additive.
    Then, for any pairwise orthogonal projectors \(\Pbf_n\in\Bcal(\ell_2(\kappa))\), we have
    \begin{multline*}
        \Braket{\Phbf_\mu\rho, \sum_n \Pbf_n} =
        \int \Braket{\rho, \sum_n\Sbf_j^*\Pbf_n\Sbf_j}\,d\mu(j) =
        \int \sum_n \Braket{\rho, \Sbf_j^*\Pbf_n\Sbf_j}\,d\mu(j) = \\ =
        \sum_n \int \Braket{\rho, \Sbf_j^*\Pbf_n\Sbf_j}\,d\mu(j) = 
        \sum_n \Braket{\Phbf_\mu\rho, \Pbf_n},
    \end{multline*}
    where:
    \begin{itemize}[nosep]
        \item the first and the last are by the definition of \(\Phbf_\mu\),
        \item the second is by $\sigma$-additivity of $\rho$,
        \item the third is by the Levi monotone convergence theorem.
    \end{itemize}
    This proves the $\sigma$-additivity of \(\Phbf_\mu\rho\).
\end{proof}

We say that a quantum channel \(\Phbf\) is \defing{singularising} if it maps initial states to singular ones;
strictly speaking, if \(\Phbf\rho\) is singular for any quantum state \(\rho\).
Cf.~the following theorem with the results in \cite{DzhenzherDzhenzherSakbaev26, Dzhenzher26-mc}, where analogous singularising quantum channels were obtained.

\begin{theorem}\label{t:sing-out}
    Let $\mu\colon2^\kappa\to[\,0,1\,]$ be a $\sigma$-additive measure that vanishes on singletons.
    Then the channel \(\Phbf_\mu\) is singularising.
\end{theorem}

\begin{proof}
    Let \(v \in S_1(\ell_2(\kappa))\) be from the unit sphere.
    If $\rho=\rho_u$ is a pure vector state, then
    \[
        \Braket{\Phbf_\mu\rho, \Pbf_v} = \int (\Sbf_j u, \Pbf_v \Sbf_ju)\,d\mu(j) =0,
    \]
    where the last equality follows since $\mu$ is $\sigma$-additive, and since \((\Sbf_j u, \Pbf_v \Sbf_ju)\neq 0\) holds only for \(j\) from the set \(\supp v - \supp u\), which is at most countable.
    
    Hence, for any normal state \(\rho\in\Sigma_n(\Bcal(\ell_2(\kappa)))\), the outcome \(\Phbf_\mu\rho\) vanishes on all finite-dimensional projectors, which means \(\Phbf_\mu\rho\) is singular.
    
    Finally, for a singular \(\rho\in\Sigma_s(\Bcal(\ell_2(\kappa)))\), we have
    \[
        \Braket{\Phbf_\mu\rho, \Pbf_v} = \int \Braket{\rho, \Pbf_{\Sbf_jv}}\,d\mu(j) =0,
    \]
    where the last equality follows since the term under the integral is zero.
\end{proof}

Theorems~\ref{t:pres-sig-add} and~\ref{t:sing-out} together claim that if \(\mu\) is a $\sigma$-additive measure that vanishes on singletons (constructed on a Ulam real-valued measurable cardinal), then \(\Phbf_\mu\) converts normal states to singular $\sigma$-additive states, which is fascinating.

Recall that \cite{Amosov-Sakbaev-2025} for any \(\rho\in\Sigma(\Bcal(\ell_2(\kappa)))\), there exists a finite non-negative finitely-additive measure \(\nu\) on \(S_1(\ell_2(\kappa))\) such that
\[
    \rho = \int \rho_u\,d\nu(u).
\]
So, there is another way to define a channel by the rule
\[
    \Phbf_\mu' \rho := \int d\nu(u)\int \Sbf_j\rho_u\Sbf_j^*\,d\mu(j).
\]

\begin{lemma}
    The quantum channels \(\Phbf_\mu\) and \(\Phbf_\mu'\) coincide.
    In particular, the quantum channel \(\Phbf_\mu'\) does not depend on the choice of $\nu$.
\end{lemma}

\begin{proof}
    Fix any \(\Abf\in\Bcal(\ell_2(\kappa))\) and consider the operator \(\Abf_\mu\) given by the Pettis integral
    \[
        \Abf_\mu = \int \Sbf_j^*\Abf\Sbf_j\,d\mu(j);
    \]
    strictly speaking, this is the action of the dual of the channel \(\Phbf_\mu\) on $\Abf$ in the Heisenberg representation.
    Then
    \[
        \Braket{\Phbf_\mu'\rho, \Abf} =
        \int d\nu(u)\int \Braket{\rho_u, \Sbf_j^*\Abf\Sbf_j}\,d\mu(j) =
        \Braket{\rho, \Abf_\mu} =
        \int \Braket{\rho, \Sbf_j^*\Abf\Sbf_j}\,d\mu(j) = \Braket{\Phbf_\mu\rho, \Abf}.
    \]
\end{proof}

If the measure \(\mu\) had not been defined on all subsets of \(\kappa\), then the channels might not have coincided; moreover, the channel \(\Phbf_\mu\) might have been undefined, as in \cite{DzhenzherDzhenzherSakbaev26}. So, the (real-valued Ulam) measurability of the cardinal is a very strong condition.

\section{Quantum channels preserving purity}\label{s:purity}

In this section, we turn our attention to the geometric and facial structure of the state space under the action of the constructed quantum channels. Specifically, we investigate the conditions under which the averaging operators preserve the purity of states, mapping pure vector states into pure quantum states. As we shall demonstrate, this property is intimately connected to the algebraic structure of the channel's support and is heavily constrained by the cardinality of the continuum $\mathfrak{c}$. By employing excision criteria for states on operator algebras, we establish structural limits on the existence of non-trivial purity-preserving singularising channels.

\begin{theorem}
    Let \(\kappa \leq \continuum\).
    Let $\mu\colon2^\kappa\to[\,0,1\,]$ be a $\sigma$-additive measure that vanishes on singletons.
    Then for any $\sigma$-additive state $\rho\in\Sigma(\Bcal(\ell_2(\kappa)))$, the outcome \(\Phbf_\mu\rho\) is not pure.
    In particular, this holds for any normal quantum state~\(\rho\in\Sigma_n(\Bcal(\ell_2(\kappa)))\), and for any vector pure state.
\end{theorem}

\begin{proof}
    By theorems~\ref{t:pres-sig-add} and~\ref{t:sing-out}, \(\Phbf_\mu\) maps $\sigma$-additive states into singular $\sigma$-additive states.
    By \cite[Lemma~5.2]{Blecher-Weaver}, if there exists a singular $\sigma$-additive pure state, then \(\kappa > \continuum\).
\end{proof}

Below, the object of our interest is whether \(\Phbf_\mu\) preserves purity regardless of $\kappa$ being greater than $\continuum$, or not.

\begin{theorem}
    Let $\mu\colon2^\kappa\to[\,0,1\,]$ be a measure.
    If $\mu$ is not two-valued, then \(\Phbf_\mu\) does not preserve the purity of any vector pure state.
\end{theorem}

\begin{proof}
    Since $\mu$ is not two-valued, there exists a decomposition \(A \sqcup B=\kappa\) such that
    \[
        0 < \mu(A),\mu(B) < 1.
    \]
    For a vector pure state \(\rho=\rho_u\), consider two states
    \[
        \rho_A := \frac{1}{\mu(A)}\int_A \Sbf_j\rho\Sbf_j^*\,d\mu(j)
        \qquad\text{and}\qquad
        \rho_B := \frac{1}{\mu(B)}\int_B \Sbf_j\rho\Sbf_j^*\,d\mu(j).
    \]
    Then
    \[
        \Phbf_\mu \rho = \mu(A)\rho_A + \mu(B)\rho_B,
    \]
    and it remains to prove that \(\rho_A\neq \rho_B\).

    Let \(\Hcal_A\) be a (closed) subspace generated by \(\Sbf_j u\) for \(j\in A\).
    Let \(\Pbf_A\) be the projector onto \(\Hcal_A\).
    Then
    \[
        \Braket{\rho_A, \Pbf_A} = \frac{1}{\mu(A)}\int_A \Braket{\rho_{\Sbf_j u}, \Pbf_A}\,d\mu(j) = \frac{1}{\mu(A)}\int_A 1\,d\mu(j) = 1.
    \]
    For $\mu$-almost every \(j \in B\), we have \(\Braket{\rho_{\Sbf_j u}, \Pbf_A}=0\).
    Hence,
    \[
        \Braket{\rho_B, \Pbf_A} = \frac{1}{\mu(B)}\int_B \Braket{\rho_{\Sbf_j u}, \Pbf_A}\,d\mu(j) = 0.
    \]
    So, \(\Braket{\rho_A, \Pbf_A}=1\neq 0 =\Braket{\rho_B, \Pbf_A}\), and thus \(\rho_A\neq\rho_B\).
\end{proof}

In order to establish whether \(\Phbf_\mu\) preserves purity for
a two-valued measure \(\mu\), we need to develop the theory of pure states. We will follow the notations from \cite{Blecher-Weaver}.


A family \(\Fcal \subset \Pcal(\Hcal)\cap\Mcal\) of projectors is called a \emph{quantum filter} if the following two properties hold:
\begin{enumerate}[label=\arabic*.]
    \item if \(\Pbf \in \Fcal\) and \(\Pbf \leq \Pbf'\in\Pcal(\Hcal)\cap\Mcal\), then \(\Pbf'\in\Fcal\),
    \item if \(\Pbf_1,\ldots,\Pbf_n\in\Fcal\), then \(\norm{\Pbf_1\ldots\Pbf_n}=1\).
\end{enumerate}
For example, for a state \(\rho\in\Sigma(\Mcal)\), the family
\[
    \Fcal_\rho := \left\{\Pbf \in\Pcal(\Hcal)\cap\Mcal \;:\;  \Braket{\rho, \Pbf}=1\right\}
\]
is a quantum filter; it is not empty since, for example, it contains the identity operator.
We will not use quantum filters except \(\Fcal_\rho\).

\begin{proposition}\label{p:qfilter}
    Let \(\rho\in\Sigma(\Mcal)\) be a quantum state on a von Neumann algebra \(\Mcal\).
    Then for any \(\Pbf\in\Fcal_\rho\) and \(\Abf\in\Mcal\) we have
    \[
        \Braket{\rho, \Abf} = \Braket{\rho, \Pbf\Abf} = \Braket{\rho, \Abf\Pbf}.
    \]
\end{proposition}

\begin{proof}
    Due to the symmetry of the equation, it is sufficient to prove the first equality.
    Let
    \[
        \Pbf^\perp = \Ibf - \Pbf
    \]
    the orthogonal projector.
    By linearity,
    \[
        \Braket{\rho, \Pbf^\perp} = \Braket{\rho, \Ibf} - \Braket{\rho, \Pbf} = 1 - 1 = 0.
    \]
    Then, by the Cauchy--Schwarz inequality \cite[Lemma~4.7]{Farah-Wofsey},
    \[
        \abs{\Braket{\rho, \Pbf^\perp\Abf}}^2 \leq \Braket{\rho, \Pbf^\perp}\Braket{\rho, \Abf^*\Abf} = 0.
    \]
    Thus
    \[
        \Braket{\rho, \Pbf^\perp\Abf} = 0,
    \]
    and so
    \[
        \Braket{\rho, \Abf} = 
        \Braket{\rho, \Pbf\Abf} + \Braket{\rho, \Pbf^\perp\Abf} = \Braket{\rho, \Pbf\Abf}.
    \]
\end{proof}

In \cite{Akemann-Anderson-Pedersen-1986}, the criterion of purity  was given in terms of excision.
We give a slightly different definition of excision, as in \cite{Blecher-Weaver}.
So, let $\rho\in\Sigma(\Mcal)$ be a quantum state on a von Neumann algebra \(\Mcal\).
For \(\Abf\in\Mcal\), we say that \defing{\(\Pbf\in\Fcal_\rho\) excises \(\rho\) for \(\Abf\)}, if
\[
    \Pbf\Abf\Pbf = \Braket{\rho, \Abf}\Pbf.
\]

\begin{lemma}[Excision Criterion]\label{l:excision}
    Let $\rho\in\Sigma(\Hcal)$ be a quantum state on a von Neumann algebra $\Mcal$.
    If for any \(\Abf\in\Mcal\) there exists a projector \(\Pbf\in\Fcal_\rho\) that excises $\rho$ for \(\Abf\), then $\rho$ is pure.
    The converse holds if $\rho$ is $\sigma$-additive.
\end{lemma}

\begin{proof}
    The ``converse'' part is exactly \cite[Lemma~4.5]{Blecher-Weaver}, so we need to prove the ``if'' part.
    Take a decomposition
    \[
        \rho = \lambda\rho_1 + (1-\lambda)\rho_2
    \]
    for some \(0 < \lambda < 1\) and \(\rho_1\neq\rho_2\).
    Fix an arbitrary $\Abf\in\Mcal$.
    Take the projector \(\Pbf\in\Fcal_\rho\) that excises $\rho$ for \(\Abf\).
    Applying the action of \(\rho_1\) on
    \[
        \Pbf\Abf\Pbf = \Braket{\rho, \Abf}\Pbf,
    \]
    we obtain
    \[
        \Braket{\rho_1, \Pbf\Abf\Pbf} = \Braket{\rho, \Abf}\Braket{\rho_1, \Pbf} = \Braket{\rho, \Abf}.
    \]
    Since
    \[
        1 = \Braket{\rho, \Pbf} = \lambda\Braket{\rho_1, \Pbf} + (1-\lambda)\Braket{\rho_2, \Pbf} \leq
        \lambda + (1-\lambda) = 1,
    \]
    we obtain that both quantum filters \(\Fcal_{\rho_1}\) and \(\Fcal_{\rho_2}\) contain \(\Pbf\).
    Then by Proposition~\ref{p:qfilter},
    \[
        \Braket{\rho_1, \Abf} = \Braket{\rho_1, \Pbf\Abf\Pbf} = \Braket{\rho, \Abf}.
    \]
    Due to the arbitrary choice of \(\Abf\), we obtain that \(\rho=\rho_1\).
    Analogously, \(\rho=\rho_2\), which concludes the proof.
\end{proof}

\begin{example}\label{ex:vec-non-pure}
    Consider the Hilbert space \(\Hcal := \Cx^2\) of qubits, and a von Neumann algebra \(\Mcal\) of diagonal matrices.
    Let \(u = \frac{1}{\sqrt{2}}\begin{pmatrix}
        1&1
    \end{pmatrix}\), and \(\rho = \rho_u\).
    There are only four projectors in \(\Mcal\), and only \(\Ibf\) of them all lies in \(\Fcal_\rho\).
    Take
    \[
        \Abf = \begin{pmatrix}
            1 & 0 \\ 0 & 0
        \end{pmatrix}.
    \]
    Then \(\Braket{\rho, \Abf} = \frac{1}{2}\), and
    \[
        \Ibf\Abf\Ibf = \Abf \neq \Braket{\rho, \Abf}\Ibf.
    \]
    Hence, by Lemma~\ref{l:excision}, the vector state \(\rho_u\) is not pure.
\end{example}

In light of the previous example, the following proposition does not seem so trivial.

\begin{proposition}\label{p:excise-pv}
    Let \(u\in S_1(\Hcal)\).
    Let \(\rho = \rho_u\in\Sigma(\Mcal)\) be a vector state, and \(\Pbf_u \in \Mcal\).
    Then $\rho_u$ is pure if and only if the projector \(\Pbf_u\) on $u$ excises \(\rho_u\) for any \(\Abf\in\Mcal\).
\end{proposition}

\begin{proof}
    The ``if'' part is by Lemma~\ref{l:excision}, so we need to prove the ``only if'' part.
    It is clear that \(\Pbf_u \in \Fcal_{\rho_u}\).
    Let $\Abf\in\Mcal$.
    So it is sufficient to prove that for any $v \in \Hcal$,
    \[
        \Pbf_u\Abf\Pbf_u v = \Braket{\rho, \Abf}\Pbf_u v.
    \]
    We have
    \[
        \Pbf_u v = (u,v)u.
    \]
    Hence
    \[
        \Pbf_u\Abf\Pbf_u v = (u,v) (u, \Abf u) u.
    \]
    The equality \(\Braket{\rho_u, \Abf} = (u, \Abf u)\) implies the result.
\end{proof}

We will need more technical criteria for purity.

\begin{lemma}[Approximative Excision]\label{l:excision-appx}
    Let \(\rho\in\Sigma(\Mcal)\) be a $\sigma$-additive state on a von Neumann algebra $\Mcal$.
    Then $\rho$ is pure if and only if for any \(\Abf\in\Mcal\) there exists a countable family of projectors \(\Pbf_n \in \Fcal_\rho\) such that
    \[
        \lim_{n\to\infty} \norm{\Pbf_n(\Abf - \Braket{\rho, \Abf}\Ibf)\Pbf_n} = 0.
    \]
\end{lemma}

\begin{proof}
    The ``only if'' part holds by Lemma~\ref{l:excision} (Excision Criterion) since $\rho$ is $\sigma$-additive, and we may take $\Pbf_n$ being the same projector excising $\rho$ for $\Abf$.
    So we need to prove the ``if'' part.
    
    Fix any $\Abf\in\Mcal$.
    Denote
    \[
        \Pbf := \bigwedge_n \Pbf_n.
    \]
    Since \(\rho\) is $\sigma$-additive, by \cite[Theorem~4.1]{Blecher-Weaver}
    \[
        \Braket{\rho, \Pbf^\perp} = \Braket{\rho, \bigvee_n\Pbf_n^\perp} = 0,
    \]
    and so \(\Pbf \in \Fcal_\rho\).
    Thus, by Lemma~\ref{l:excision} (Excision Criterion), it is sufficient to prove that
    \[
        \Pbf\Abf\Pbf = \Braket{\rho, \Abf}\Pbf.
    \]
    Since \(\Pbf = \Pbf_n\Pbf = \Pbf\Pbf_n\), we have
    \[
        \norm{\Pbf(\Abf - \Braket{\rho, \Abf}\Ibf)\Pbf} \leq \norm{\Pbf_n(\Abf - \Braket{\rho, \Abf}\Ibf)\Pbf_n} \xrightarrow[n\to\infty]{} 0,
    \]
    which concludes the proof.
\end{proof}

Now we focus our attention on $\sigma$-additive two-valued measures \(2^\kappa\to\{0,1\}\) vanishing on singletons; they exist \cite{Jech-p1-ch10} if and only if $\kappa$ is a Ulam measurable cardinal.
Recall that there is a bijection between all two-valued measures \(2^\kappa\to\{0,1\}\) and ultrafilters on~\(\kappa\).
This means that to each two-valued measure \(\mu\colon2^\kappa\to\{0,1\}\) corresponds an ultrafilter \(\Ucal_\mu\) of sets of measure~$1$.
It is clear that
\begin{itemize}[nosep]
    \item $\mu$ is $\sigma$-additive if and only if $\Ucal_\mu$ is $\sigma$-complete (which means that \(\Ucal_\mu\) is closed under countable intersections), and
    \item $\mu$ vanishes on singletons if and only if $\Ucal_\mu$ is non-principal.
\end{itemize}
There exists an even stronger connection between such measures and ultrafilters. The following result is probably well known, but we do not know the reference for it.

\begin{lemma}\label{l:int-ultralim}
    Let \(\mu\colon2^\kappa\to\{0,1\}\) be a two-valued measure.
    Then for any bounded function \(f\colon\kappa\to\Cx\), the integral over $\mu$ and the limit along $\Ucal_\mu$ coincide, meaning
    \[
        \lim_{j\to \Ucal_\mu} f(j) = \int f(j)\,d\mu(j).
    \]
\end{lemma}

\begin{proof}
    It is sufficient to prove the case \(f\colon\kappa\to\R\).
    Since $f$ is bounded, there exists a unique
    \[
        L := \lim_{j\to \Ucal_\mu} f(j) \in \R.
    \]
    By definition of $L$, for any $\varepsilon>0$ holds
    \[
        A_\varepsilon := \{j\in\kappa:L-\varepsilon < f(j) < L+\varepsilon\} \in \Ucal_\mu.
    \]
    The latter means that \(\mu(A_\varepsilon)=1\).
    So,
    \[
        L-\varepsilon<\int_{A_\varepsilon} f(j)\,d\mu(j) < L+\varepsilon.
    \]
    The proof is concluded by the fact that
    \[
        \int f(j)\,d\mu(j) = \int_{A_\varepsilon} f(j)\,d\mu(j) + \int_{\kappa\setminus A_\varepsilon} f(j)\,d\mu(j) = \int_{A_\varepsilon} f(j)\,d\mu(j)
    \]
    and the arbitrary choice of $\varepsilon$.
\end{proof}

\begin{theorem}\label{t:purity}
    Let $\mu\colon2^\kappa\to\{0,1\}$ be a $\sigma$-additive two-valued measure.
    Then \(\Phbf_\mu\) preserves the purity of pure vector states constructed on basis vectors $e_i$.
\end{theorem}

\begin{proof}
    It is sufficient to prove the theorem for a pure vector state \(\rho = \rho_{e_0}\), since for others they differ only by shift.
    Take an ultrafilter $\Ucal_\mu$.
    Then, by Lemma~\ref{l:int-ultralim},
    \[
        \Phbf_\mu \rho = \lim_{j\to\Ucal_\mu} \Sbf_j\rho_{e_0}\Sbf_j^* =
        \lim_{j\to\Ucal_\mu} \rho_{e_j},
    \]
    where the ultralimit is in the ultraweak sense.
    By \cite[Theorem~5.1]{Blecher-Weaver}, it is sufficient to prove that \(\Phbf_\mu \rho\) is pure on the diagonal subalgebra.
    Fix any diagonal \(\Abf\in\Bcal(\Hcal)\) and \(\varepsilon>0\).
    In order to apply Lemma~\ref{l:excision-appx} (Approximative Excision), we need to construct the excision projector.
    Let
    \[
        \Abf = \sum_{j\in\kappa} a_j\Pbf_{e_j}.
    \]
    By definition of ultralimit,
    \[
        X_\varepsilon := \left\{j\in \kappa : \abs{a_j - \Braket{\Phbf_\mu\rho,\Abf}} < \varepsilon\right\} \in \Ucal_\mu.
    \]
    Thus, we are able to define the projector
    \[
        \Pbf_\varepsilon := \sum_{j \in X_\varepsilon} \Pbf_{e_j}.
    \]
    Since \(X_\varepsilon \in \Ucal_\mu\),
    \[
        \Braket{\Phbf_\mu\rho, \Pbf_\varepsilon} =
        \lim_{j\to\Ucal_\mu} (e_j, \Pbf_\varepsilon e_j) =
        \lim_{X_\varepsilon \ni j\to\Ucal_\mu} (e_j, \Pbf_\varepsilon e_j) =
        \lim_{X_\varepsilon \ni j\to\Ucal_\mu} 1 = 1,
    \]
    and so \(\Pbf_\varepsilon \in \Fcal_{\Phbf_\mu \rho}\).

    Denote
    \[
        \Dbf := \Abf - \Braket{\Phbf_\mu\rho, \Abf}\Ibf =
        \sum_{i\in\kappa} d_i\Pbf_{e_i},
    \]
    where
    \[
        d_i := a_i-\Braket{\Phbf_\mu\rho, \Abf}.
    \]
    It is clear that
    \[
        \Pbf_\varepsilon\Dbf\Pbf_\varepsilon =
        \sum_{i\in\kappa}\sum_{j\in Y_\varepsilon}\sum_{k\in X_\varepsilon} d_i\Pbf_{e_j}\Pbf_{e_i}\Pbf_{e_k} =
        \sum_{j\in X_\varepsilon} d_j\Pbf_{e_j}.
    \]
    Since \(\abs{d_j}< \varepsilon\) for \(j\in X_\varepsilon\), 
    so is \(\norm{\Pbf_\varepsilon\Dbf\Pbf_\varepsilon}\leq\varepsilon\).
    Thus, Lemma~\ref{l:excision-appx} (Approximative Excision) concludes the proof.
\end{proof}

\begin{example}\label{ex:purity}
    Theorem~\ref{t:purity} does not have to hold for arbitrary pure vector state.
    Indeed, let \(\kappa=\R\), and take \(u = \frac{e_0 + e_1}{\sqrt{2}}\).
    Let $+$ be the usual summation.
    For an ultrafilter \(\Ucal\) on \(\R\), define
    \[
        \rho_\Ucal := \lim_{j\to\Ucal_\mu} \rho_{e_j}.
    \]
    For any diagonal operator $\Dbf\in\Bcal(\ell_2(\R))$,
    \[
        \Braket{\rho_u, \Dbf} = \frac{1}{2}(e_0,\Dbf e_0) + \frac{1}{2}(e_1,\Dbf e_1).
    \]
    Then
    \[
        \Phbf_\mu\rho_u = \frac{1}{2}\rho_{\Ucal_\mu} + \frac{1}{2}\rho_{\Ucal_\mu - 1}.
    \]
    It remains to notice that $\Ucal_\mu$ is not translation invariant: indeed, for example, if
    \[
        X := \bigsqcup_{n=-\infty}^\infty [\,2n, 2n+2)
    \]
    lies in \(\Ucal_\mu\), then \(X+1 = \R\setminus X\) does not, and vise versa.
    So, taking \(\Dbf\) to be the projector on \(X\), we see that
    \[
        \Braket{\rho_{\Ucal_\mu},\Dbf} = 1
        \quad\text{and}\quad
        \Braket{\rho_{\Ucal_\mu-1},\Dbf} = 0,
    \]
    or vise versa. So, \(\rho_{\Ucal_\mu}\neq\rho_{\Ucal_\mu-1}\), and thus \(\Phbf_\mu\rho_u\) is not pure.
\end{example}

\section{Dynamics on unitary groups}\label{s:unitary}

In the previous sections, our focus was on the group of translations \(\{\Sbf_j\}_{j\in \kappa}\).
In this section, we extend our considerations from the specific shift operators to a broader class of unitary representations of an abstract group \((G, \cdot)\).
So, let \(\Ucal(\Hcal)\) be the group of unitary operators on a Hilbert space \(\Hcal\).
Let \((G,\cdot)\) be some other group; we will omit the \(\cdot\) operation.
Denote by $2^G$ the set of all subsets of the group~$G$.
Let \(\mu\colon 2^{G}\to[\,0,1\,]\) be a finitely-additive measure on this group.
Let \(\pi\colon G\to\Ucal(\Hcal)\subset GL(\Hcal)\) be a representation of this group.
Define the quantum channel \(\Qbf_{\mu,\pi}\colon \Sigma(\Bcal(\Hcal))\to\Sigma(\Bcal(\Hcal))\) by the Pettis integral
\[
    \Qbf_{\mu,\pi} \rho := \int \pi(g)\rho\pi(g)^*\,d\mu(g).
\]
Note that the channel \(\Phbf_\mu\) is a particular case of the channel \(\Qbf_{\mu,\pi}\) for a measure $\mu$ concentrated on the family of shift operators.
In \cite[Theorem~3.1]{Amosov-Sakbaev-2025}, it was shown that if $\Hcal$ is separable and \(\mu\) is $\sigma$-additive, then \(\Qbf_{\mu,\pi}\) preserves normality; see also \cite[Theorem~6.3]{Amosov-Sakbaev-2025} for the channel constructed by the representation over a topological group.

In order to prevent any set-theoretic or measure-theoretic ambiguities, we explicitly distinguish the domain and properties of the measure \(\mu \) as follows.
\begin{itemize}[nosep]
    \item When analysing the preservation of \(\sigma\)-additivity, \(\mu\) is assumed to be a \(\sigma\)-additive measure defined on the full power set \(2^{G}\). Such $\sigma$-additive measures always exist: we may always take convex combinations of delta measures (note that we do not require here that $\mu$ vanishes on singletons or is left-invariant).
    Note that if \(\kappa\) carries a \(\sigma\)-additive measure \(\mu\) vanishing on singletons, such a measure can be naturally extended to a \(\sigma\)-additive measure on the full power set of any larger group \(G\) containing \(\kappa\) (such as \(\mathcal{U}(\ell^2(\kappa))\)) by setting \(\mu_G(E) = \mu(E \cap \kappa)\) for \(E \subset G\) (in other words, if \(\kappa\) is Ulam measurable, then so is \(\Ucal(\ell_2(\kappa))\)). This ensures the existence of \(\sigma\)-additive measures on \(2^{G}\) required for Theorem~\ref{t:sigma-additive-on-group}, although such extended measures are supported on the subset \(\kappa \subset G\). See \cite[Chapter~IX]{Kuratowski-1967set} for details.
    \item When analysing the fixed points under the action of the group and using left-invariance, \(\mu\) is assumed to be a finitely-additive left-invariant measure defined on the full power set \(2^{G}\). We explicitly emphasise that \(\mu\) in Theorem~\ref{t:left-invarinant-iter} is not required to be \(\sigma\)-additive, which perfectly aligns with the classical non-existence results for left-invariant $\sigma$-additive measures on power sets.
\end{itemize}

\begin{theorem}\label{t:sigma-additive-on-group}
    Let \(\Hcal\) be a Hilbert space.
    Let \(\mu\colon 2^G\to[\,0,1\,]\) be a $\sigma$-additive measure on a group~$G$.
    Let \(\pi\colon G\to\Ucal(\Hcal)\).
    Then the channel \(\Qbf_{\mu,\pi}\) preserves $\sigma$-additivity.
\end{theorem}

\begin{proof}
    The proof is analogous to the proof of Theorem~\ref{t:pres-sig-add}.
\end{proof}

\begin{example}
    Here we show that the analogue of Theorem~\ref{t:sing-out} may not hold for \(\Qbf_{\mu,\pi}\) for some representation $\pi$.
    Indeed, let \(\pi(g):=\Ibf\).
    Then
    \[
        \Braket{\Qbf_{\mu,\pi} \rho,\Abf} = \int \Braket{\rho,\Abf}\,d\mu(g) = \Braket{\rho,\Abf}\mu(G) = \Braket{\rho,\Abf},
    \]
    which means that \(\Qbf_{\mu,\pi} \rho = \rho\) regardless of $\mu$, and there is no singularisation.
\end{example}

As shown in Example~\ref{ex:purity}, the channel \(\Qbf_{\mu,\pi}\) does not necessarily preserve purity.
However, this quantum channel is interesting in another way.

The measure $\mu$ on the group \(G\) is said to be \defing{left-invariant} if \(\mu(E) = \mu(g E)\) for any \(g\in G\) and \(E\subset G\).
See, for example, \cite{Rudin-1991-fa} for the introduction to Haar measures.
In the case of discrete groups, the left-invariant measure is the counting measure, which is obviously $\sigma$-additive; this is the case when \(\Qbf_{\mu,\pi}\) is given by the Kraus representation.
Below, speaking of left-invariant measures, we do not presume that they are $\sigma$-additive.

We say that a state \(\rho\in\Bcal(\Hcal)\) is a \emph{fixed point of any unitary evolution} if \(\rho=\Ubf\rho\Ubf^*\) for any \(\Ubf\in\Ucal(\Hcal)\).

\begin{theorem}\label{t:left-invarinant-iter}
    Let \(\Hcal\) be an infinite-dimensional Hilbert space.
    Let \(\mu\colon 2^G\to[\,0,1\,]\) be left-invariant.
    Then \(\Qbf_{\mu,\pi}^2=\Qbf_{\mu,\pi}\).
\end{theorem}

\begin{proof}
    The theorem follows since for any \(\rho\in\Sigma(\Bcal(\Hcal))\), \(\Abf\in\Bcal(\Hcal)\), and \(g\in G\),
    \begin{multline*}
        \Braket{\pi(g)\Qbf_{\mu,\pi}\rho\pi(g)^*, \Abf} = \int\Braket{\rho, \pi(h)^*\pi(g)^*\Abf\pi(g)\pi(h) }\,d\mu(h) = \\ =
        \int\Braket{\rho, \pi(t)^*\Abf\pi(t)}\,d\mu(g^{-1}t) =
        \int\Braket{\rho, \pi(t)^*\Abf\pi(t)}\,d\mu(t) = \Braket{\Qbf_{\mu,\pi}\rho, \Abf}.
    \end{multline*}
\end{proof}

Thus, unlike the results in \cite{DzhenzherDzhenzherSakbaev26}, where a semigroup of singularising quantum channels was obtained for \(\mu_s*\mu_t=\mu_{s+t}\), here the quantum channels \(\Qbf_{\mu_t}\) will not form a semigroup.

We now turn our attention to the dynamics of this quantum channel.

Recall that for two finitely-additive measures \(\mu,\nu\colon2^G\to[\,0,1\,]\), their \emph{convolution} \(\mu*\nu\colon2^G\to[\,0,1\,]\) is defined by
\[
    \mu*\nu(E) := \int \nu\left(g^{-1}E\right)\,d\mu(g).
\]
In other words,
\[
    \int f(t)\,d\mu*\nu(t) = \int d\mu(g) \int d\nu(h)\,f(gh)
\]
for any bounded \(f\colon G\to\Cx\).

\begin{proposition}\label{p:din-conv}
    For any \(\mu\colon 2^G\to [\,0,1\,]\),
    \[
        \Qbf_{\mu,\pi}^n = \Qbf_{\mu^{*n}}.
    \]
\end{proposition}

\begin{proof}
    This is since for any state $\rho$,
    \[
        \Qbf_{\mu,\pi}^n\rho = \int d\mu(g_1)\int \ldots \int d\mu(g_n)\,\pi(g_1\ldots g_n)\rho\pi(g_1\ldots g_n)^* =
        \int \pi(g)\rho\pi(g)^*\,d\mu^{*n}(g).
    \]
\end{proof}

For a measure \(\mu\) on \(G\) and some \(h\in G\), define the left shifted measure \(L_h\mu\colon 2^G\to [\,0,1\,]\) by
\[
    L_h\mu(E) := \mu(h E).
\]

\begin{theorem}\label{t:as-left-inv}
    Let $\Hcal$ be a Hilbert space.
    Let \(\mu\colon 2^G\to [\,0,1\,]\) be a finitely-additive measure.
    Let \(\pi\colon G\to\Ucal(\Hcal)\).
    Suppose that for any state \(\rho\in\Sigma(\Bcal(\Hcal))\), iterations \(\Qbf_{\mu,\pi}^n\rho\) of quantum channels weakly-$*$ converge
    to a state (depending on $\rho$), which is a fixed point of any unitary evolution.
    Then
    \[
        \tau\lim_{n\to\infty}\left(\mu^{*n} - L_h(\mu^{*n})\right) = 0
    \]
    for any \(h\in G\), where $\tau$ is the weak topology induced by functions
    \[
        g \mapsto \Braket{\rho, \pi(g)^*\Abf\pi(g)},
        \qquad g\in G,
        \qquad \Abf\in\Bcal(\Hcal),
        \quad  \rho\in\Sigma(\Bcal(\Hcal)).
    \]
\end{theorem}

\begin{proof}
    Take any function
    \[
        f\colon g \mapsto \Braket{\rho, \pi(g)^*\Abf\pi(g)},
    \]
    given by some \(\Abf\in\Bcal(\Hcal)\) and \(\rho\in\Sigma(\Bcal(\Hcal))\).
    Let \(\rho_\infty\) be a weak-$*$ limit of \(\Qbf_{\mu,\pi}^n\rho\).
    Take any \(h\in G\).
    Then, by Proposition~\ref{p:din-conv},
    \[
        \int f\,d\mu^{*n} = \Braket{\Qbf_{\mu,\pi}^n\rho, \Abf} \xrightarrow[n\to\infty]{} \Braket{\rho_\infty, \Abf}
    \]
    and
    \[
        \int f\,dL_h\mu^{*n} = \Braket{\Qbf_{\mu,\pi}^n\rho, \pi(h)^*\Abf\pi(h)} \xrightarrow[n\to\infty]{} \Braket{\rho_\infty, \pi(h)^*\Abf\pi(h)}.
    \]
    The equality \(\Braket{\rho_\infty, \Abf} = \Braket{\rho_\infty, \pi(h)^*\Abf\pi(h)}\), which holds since \(\rho_\infty\) is a fixed point of any unitary evolution, concludes the proof.
\end{proof}

As Theorem~\ref{t:as-left-inv} shows, if \(\mu=\mu*\mu\) is not left-invariant, then there can be no convergence to states that are fixed points of any unitary evolution.

For the final touch, we consider the dynamics of quantum channels.
The following results on convergence may not hold for iterations of quantum channels, so we prove them for the Ces\'{a}ro averages.
We will rely heavily on the following result of Yosida and Kakutani \cite[Theorem~1]{Yosida-Kakutani}.

\begin{theorem}[Yosida--Kakutani]\label{t:Yosida-Kakutani}
    Let \(T\colon B\to B\) be a bounded linear operator on a Banach space $B$.
    Suppose that there exists \(C>0\) such that for any $n\geq 1$ we have \(\norm{T^n} \leq C\).
    Suppose that for any \(x\in B\), the sequence of the Ces\'{a}ro averages
    \[
        x_n := \frac{1}{n}\left(T+\ldots + T^n\right)x
    \]
    contains a subsequence that converges weakly to some \(\overline{x}\in B\).
    Then \(x_n\) converges strongly to \(\overline{x}\), and the operator
    \[
        \overline{T}\colon x\mapsto \overline{x}
    \]
    is a bounded linear operator \(B\to B\) such that
    \[
        \overline{T}T = T\overline{T} = \overline{T}^2 = \overline{T}.
    \]
\end{theorem}

Below by the quantum channel \(\Phbf\colon\Sigma_n(\Mcal)\to\Sigma_n(\Mcal)\) we mean the linear positive map \(\Phbf\colon \Mcal_*\to\Mcal_*\) such that \(\Phbf(\Sigma_n(\Mcal))\subset\Sigma_n(\Mcal)\).

\begin{theorem}\label{t:norm-cesaro}
    Let \(\Mcal\) be a von Neumann algebra on a Hilbert space $\Hcal$.
    Let \(\Phbf\colon\Sigma_n(\Mcal)\to\Sigma_n(\Mcal)\) be a quantum channel.
    Then the Ces\'{a}ro averages
    \[
        \frac{1}{n}\left(\Phbf + \ldots + \Phbf^n\right)
    \]
    converge in strong operator topology to a quantum channel \(\overline{\Phbf}\) such that
    \[
        \overline{\Phbf}\Phbf = \Phbf\overline{\Phbf} = \overline{\Phbf}^2 = \overline{\Phbf}.
    \]
\end{theorem}

\begin{proof}
    Fix any \(\rho\in\Sigma_n(\Mcal)\).
    Since \(\norm{\Phbf}\leq 1\), the sequence of
    \[
        \rho_n := \frac{1}{n}\left(\Phbf + \ldots + \Phbf^n\right)\rho \in \Sigma_n(\Mcal)
    \]
    is well defined.
    Since \(\Sigma_n(\Mcal)\) is weakly compact, the sequence \(\{\rho_n\}\) contains a subsequence that converges weakly to some \(\overline{\rho}\in \Sigma(\Mcal)\).
    Since \(\norm{\Phbf} \leq 1\), it remains to apply Theorem~\ref{t:Yosida-Kakutani} for \(B=\Mcal_*\), \(T=\Phbf\), $x=\rho$, and \(x_n=\rho_n\).
    This shows that \(\rho_n\) converges to \(\overline{\rho} =: \overline{\Phbf}\rho\) strongly.
\end{proof}

By Theorem~\ref{t:norm-cesaro} and \cite[Theorem~3.1]{Amosov-Sakbaev-2025}, for a separable Hilbert space \(\Hcal\) and a $\sigma$-additive measure \(\mu\) on \(2^{\Ucal(\Hcal)}\) itself (taking the representation \(\pi\colon\Ubf\mapsto\Ubf\)), the Ces\'{a}ro averages of \(\Qbf_{\mu,\pi}\colon\Sigma_n(\Bcal(\Hcal))\to\Sigma_n(\Bcal(\Hcal))\) converge in SOT to some quantum channel \(\overline{\Phbf}\colon\Sigma_n(\Bcal(\Hcal))\to\Sigma_n(\Bcal(\Hcal))\).
It is interesting to find out whether \(\overline{\Phbf}\) itself can be given as the Pettis integral over some measure, which probably should be left-invariant.

\section{Conclusion}

In this paper, we have investigated the structural and dynamical properties of quantum channels $\Phbf_\mu$ and $\Qbf_{\mu,\pi}$ induced by averaging procedures via the Pettis integral. By establishing a bridge between operator algebra theory and the foundational properties of Ulam real-valued measurable cardinals, we demonstrated how the analytical behaviour of the underlying measure directly dictates the topological nature of the resulting quantum states. Specifically, we provided a complete characterisation of channels that simultaneously preserve $\sigma$-additivity while enforcing a strict singularisation of any incoming state.
Furthermore, we extended these constructions to general groups of unitary operators, characterising their fixed points and analysing the asymptotic behaviour of their iterations.

The framework developed herein highlights the subtle geometry of singular $\sigma$-additive states, which naturally emerge at the intersection of non-separable Hilbert spaces and non-trivial set-theoretic assumptions. While the ergodic properties of normal states under such channels can be successfully resolved via classical mean ergodic theorems, the singularised dynamics pose significantly deeper challenges.
As it was mentioned for a separable Hilbert space \(\Hcal\) and a $\sigma$-additive measure \(\mu\colon2^{\Ucal(\Hcal)}\to[\,0,1\,]\), the Ces\'{a}ro averages of \(\Qbf_{\mu,\pi}\colon\Sigma_n(\Bcal(\Hcal))\to\Sigma_n(\Bcal(\Hcal))\) converge in SOT to some quantum channel \(\Sigma_n(\Bcal(\Hcal))\to\Sigma_n(\Bcal(\Hcal))\).
It is natural to examine the analogous results for singularising quantum channels.
However, it is much harder to achieve such results, since we cannot apply the Yosida--Kakutani result (Theorem~\ref{t:Yosida-Kakutani}).
There are analogous results for quasi-compact operators; 
see, for example, \cite[Uniform ergodic theory]{Dunford-Schwartz-vol1}.
Unfortunately, we have been unable to establish the quasi-compactness of singularising quantum channels \(\Qbf_{\mu,\pi}\) in any  particular noteworthy cases.

\printbibliography

\end{document}